\documentclass[runningheads]{llncs}
\usepackage{xspace}
\usepackage{subfiles}
\usepackage{amsfonts}
\usepackage{amsmath,amssymb,amsfonts}
\usepackage{thmtools}
\usepackage{thm-restate}
\usepackage{multirow,array}
\usepackage{booktabs}
\usepackage[T1]{fontenc}
\usepackage{graphicx}
\usepackage[boxruled]{algorithm2e}
\usepackage{dsfont}
\usepackage[dvipsnames]{xcolor}
\usepackage{mathtools}
\usepackage[normalem]{ulem}
\usepackage{hyperref}
\usepackage{cleveref}
\usepackage{comment}

\long\def\todo#1 { {\bf TODO:} [{\color{gray} #1}] }

\newcommand{\N}{\mathcal{N}}
\newcommand{\Nh}{\mathcal{N}_h}

\newcommand{\vi}{v_i}

\newcommand{\threshold}{t}

\newcommand{\Tiln}{\Tilde{n}}
\newcommand{\fn}{f(\Tilde{n})}
\newcommand{\subg}{\Gamma^{sub}}
\newcommand{\subs}{\Tilde{\sigma}}
\newcommand{\coin}{\mathcal{C}}

\newcommand{\PoW}{Proof-of-Work\xspace}
\newcommand{\PoS}{Proof-of-Stake\xspace}

\newcommand{\power}{power\xspace}
\newcommand{\powert}{power threshold\xspace}
\newcommand{\cns}{nodes\xspace}

\newcommand{\Sys}{\ensuremath{S}\xspace}

\graphicspath{{./figures/}}

\title{Rational Censorship Attack:\\
Breaking Blockchain with a Blackboard}
 \author{Michelle Yeo\inst{1} \and
 Haoqian Zhang\inst{2}
 }
 \institute{National University of Singapore\\
 \email{mxyeo@nus.edu.sg} \\ \and
 École Polytechnique Fédérale de Lausanne \\
  \email{haoqian.zhang@epfl.ch}
 }
\date{}

\begin{document}
\maketitle

\begin{abstract}
Censorship resilience is a fundamental assumption underlying the security of blockchain protocols. Additionally, the analysis of blockchain security from an economic and game theoretic perspective has been growing in popularity in recent years.
In this work, we present a surprising rational censorship attack on blockchain censorship resilience when we adopt the analysis of blockchain security from a game theoretic lens and assume all users are rational. 
In our attack, a colluding group with sufficient voting power censors the remainder nodes such that the group alone can gain all the rewards from maintaining the blockchain.
We show that if nodes are rational, coordinating this attack just requires a public read and write \emph{blackboard} and we formally model the attack using a game theoretic framework.
Furthermore, we note that to ensure the success of the attack, nodes need to know the total true voting power held by the colluding group.
We prove that the strategy to join the rational censorship attack and also for nodes to honestly declare their power is a subgame perfect equilibrium in the corresponding extensive form game induced by our attack.
Finally, we discuss the implications of the attack on blockchain users and protocol designers as well as some potential countermeasures.

\end{abstract}

\section{Introduction}
Blockchains like Bitcoin and Ethereum allow users to transact with others in a decentralized fashion without going through a central entity like a bank. A key security goal of such decentralized systems is censorship resilience: any user should in theory be able to participate in both consensus as well as the contribution of financial transactions to the blockchain.

Although transaction-level censorship has been well studied with several defensive measures proposed~\cite{DBLP:conf/ccs/AbrahamPY20,AlposACY23,BaumDF21,CachinKPS01,ChoiATB23,KostiainenGK22,LenstraW15,MillerXCSS16}, the analysis of consensus-level censorship attacks and corresponding mitigation is a lot more sparse~\cite{Bonneau16,mccorry2018smart,feather}. Indeed, two commonly proposed reasons why mitigating consensus-level censorship is not such a pressing issue is, firstly, that the amount of cost involved to coordinate such an attack is typically substantial. Indeed, an attacker not only has to gather sufficient power to launch the attack, but also has to find ways to coordinate strategies with other attacking nodes~\cite{51}. Secondly, it has also been argued that the short-term monetary benefits of attacking consensus on the blockchain could be outweighed by the long-term consequences of the attack, which includes a potential drop in the value of the underlying cryptocurrency resulting from the reduced trust in the blockchain~\cite{BahraniW23,ford2019rationality}.

On an orthogonal direction of research, recent years have seen a shift in the analysis of blockchain security from a cryptographic or worst-case perspective to an economic or game theoretic perspective (see e.g.,~\cite{BadertscherLZ21,eyal2018majority,LeshnoPS23,menipool,SchrijversBBR16} and more). 
In contrast to the classic cryptographic analysis where users are classified into the binary categories of byzantine or honest, an economic or game theoretic analyses typically assume all users in the blockchain are rational, that is, they are profit-maximizing. 
Under this perspective, a given system with rational users is shown to insecure if there exists attacks such that the revenue of the attack outweighs the cost of launching the attack.
The shift from cryptographic to rational analysis in blockchains is mainly motivated by the argument that in real world systems users might be motivated by economic reasons instead of purely byzantine ones to deviate from honest behavior. 
Hence, not only does this harmonizes well with blockchain security analysis as blockchains typically contain financial data, but rational analysis also allows for a more subtle perspective on blockchain security, where an argument for the security of the blockchain could still be made under rational arguments when the blockchain would be deemed insecure under cryptographic ones.

Interestingly, adopting the rational perspective of blockchain security could lead to arguments weakening the censorship resilience of blockchains. For instance, in~\cite{ford2019rationality} the authors point out that assuming blockchains exist in a larger financial ecosystem, there could be rational reasons to attack some blockchain system if one short sells the underlying cryptocurrency.
Thus, although the value of the cryptocurrency could diminish from attacks against the blockchain, a user that short sells sufficient amount of the currency would still make an overall profit in the larger financial ecosystem.
This weakens the aforementioned second argument for consensus-level blockchain censorship resistance, i.e., that the short term monetary benefits of censorship attacks are overwhelmed by their long-term financial impact.

Motivated by this observation, we focus on the first argument for consensus-level censorship resistance, i.e., the substantial cost involved in launching censorship attacks, and show that adopting the rational perspective makes it extremely cheap to launch a majority attack on the blockchain, which we term the \emph{rational censorship attack}. 
The rational censorship attack involves users colluding together and forming a coalition that controls sufficient voting power. Then, the colluding group simply ignores messages from the remaining excluded users, effectively removing them from participating in blockchain consensus and hence censoring them.
The surprising novelty of our attack lies in its simplicity. 
Indeed, coordinating our attack does not involve on any form of private communication between users in the blockchain except for public messages posted on a public blackboard with read and write access. In fact, we show that a single message posted on the blackboard is sufficient to initiate the attack. We also show that the knowledge that all other users will attempt to join the attacking coalition is sufficient incentive for any user to also join the coalition and attack the blockchain.


\subsection{Contributions}
Our contributions are summarized as follows:

\begin{itemize}
    \item We present a novel simple rational censorship attack on blockchains that guarantees a larger payoff for attacking nodes 
    with communication only via a public read-and-write blackboard.
    \item We formally model our attack using a game theoretic framework and prove that attacking the blockchain is a subgame perfect equilibrium in the extensive form game induced by our attack (\Cref{thm:spe} in~\Cref{sec:analysis}).
    \item We further prove that our attack strategy induces truthfulness among rational nodes, that is, nodes are incentivized to declare their true voting power for the attack.
    \item We present two plausible countermeasures to the attack in~\Cref{sec:deterrences} and discuss their effectiveness.
    \item We discuss the implications of the attack on blockchain users and protocol designers in~\Cref{sec:discussion}.
\end{itemize}


\subsection{Related work}

\textbf{Attacks on blockchain security leveraging rationality.}
The earliest work that exposes vulnerabilities in blockchain security by leveraging rationality arguments is selfish mining by Eyal and Sirer, demonstrating that the
Bitcoin protocol is not incentive-compatible~\cite{eyal2018majority}. 
Specifically, the work shows that even if the adversary controls
less than 50\% of the hashing power in Bitcoin, 
they can still launch a selfish mining attack to reap a larger fraction of rewards compared to the honest strategy of following the stipulated protocol.

Following the initial selfish mining attack analysis, 
several follow-up works present other attacks that expose further vulnerabilities behind the incentive mechanisms grounding blockchain protocols,
for instance whale attacks~\cite{liao2017incentivizing},
block withholding~\cite{Eyal2015},
stubborn mining~\cite{nayak2016stubborn},
transaction withholding~\cite{Babaioff2012},
empty block mining~\cite{Houy2014}, and
fork after withholding~\cite{Know2017}.

Ford and Böhme~\cite{ford2019rationality} proposed
a general attack on the rationality assumptions underlying blockchains and decentralized systems,
arguing that the rationality assumption is self-defeating
when analyzing the target system in the context of a larger ecosystem.
Zhang et al. presented a general rational attack on rationality by leveraging an out-of-band channel that incentivizes nodes to collude against the honest protocol~\cite{zhang2024zeroauction}.
Although we do not analyze the same problem in our work, we use a similar argument to justify why users might want to even attack the blockchain system they are a part of.


\textbf{Blockchain censorship.}
Blockchain censorship can occur through either a centralized/top-down approach~\cite{treasurymixer,treasurytornado,ofac} (e.g., mandated from a country's government), or through malicious attacks from blockchain users aiming to target other users in the system. 
While our work focuses on an attack that stems from the latter category, we note that~\cite{censorship} provides a comprehensive overview and analysis of government-imposed blockchain censorship.

Censorship stemming from malicious users can be further divided into (1) \emph{transaction-level} censorship, and (2) \emph{consensus-level} censorship.
In the former, certain transactions (for instance involving target entities) are deliberately excluded, and this is typically done by block miners to increase their miner extractable value (MEV)~\cite{censorship}.
Mitigation strategies include commit-and-reveal solution~\cite{ChoiATB23,LenstraW15}, encrypting transactions~\cite{zhang_et_al:LIPIcs.AFT.2023.3,CachinKPS01,KostiainenGK22,MillerXCSS16}, and multiparty computation (MPC)~\cite{DBLP:conf/ccs/AbrahamPY20,BaumDF21}.
Consensus-level censorship (and the focus of our work) 
excludes the target nodes completely from consensus~\cite{censorship},
thus subsuming the transaction-level censorship for the transactions and blocks from the excluded nodes.
Examples of consensus level censorship include \emph{feather fork} attacks~\cite{feather,WinzerHF19} and smart contract bribing attacks~\cite{Bonneau16,judmayer2021sok,McCorryHM18}.
Much like our censorship attack, feather fork attacks involve a public announcement to boycott a certain block from mined by some user or containing certain transactions. 
The public knowledge that certain blocks will be censored lowers the effective mining rate of other miners that mine on the victim block, which gives them some incentive to boycott these blocks.
The threat of feather forking has also been used in conjunction with bribing attacks~\cite{AvarikiotiKLM24bribe,WinzerHF19,HaoqianZhangBreakingBlockchainRationality}. These attacks typically involve a considerable amount of deposit as bribe and thus our work can be seen as an alternative feather forking solution that can be used in conjunction to their attacks to implement a \emph{costless} bribe.
Additionally, our work provides stronger guarantees to miners compared to typical feather fork attacks by \emph{ensuring a higher mining rate} to miners that join the attack coalition.
Smart contract bribing attacks employ smart contracts to automate the payment of rewards in exchange for proof of performing certain mining strategies. 
In contrast, our proposed attack does not require any bribes or external incentives to switch mining strategies, and is thus less costly to launch.
Rather, the existence of a single attack contract on the blockchain serves as incentive in itself to induce miners to join the attack.

Our work is also related to blockchain \emph{eclipse attacks}~\cite{HeilmanKZG15,TochnerZ20}. 
These attacks target the broadcast layer of the blockchain protocol and aim to control all in and outgoing connections of a subset of nodes, subsequently manipulating the victim nodes' view of the blockchain. 
The attackers can consequently leverage the power of the victim nodes to engineer a split in mining power, selfish mining, block races, and other attacks.


\textbf{Undetectable attacks.} Undetectability as an adversarial objective was first proposed in the context of MPC by Aumann and Lindell~\cite{AumannL07covert} under the notion of covert security. 
The original notion of covert security considers adversaries that choose to deviate from the protocol but also want to minimize the probability that they get caught, and serves as a useful framework to model adversarial behavior in realistic commercial settings~\cite{AumannL07covert,FaustHKS22}.
In the blockchain setting, undetectability objectives were proposed and analyzed for selfish mining~\cite{BahraniW23} and cryptographic self-selection~\cite{CaiLWZ24css}. Our work continues this line of research and incorporates detectability penalties into our cost model for our attack.

\section{Rational censorship attack}
\label{sec:attack}
We now present our rational censorship attack and game theoretic analysis which formally proves the strategy consisting of all players joining the attack coalition is a subgame perfect equilibrium in the underlying game induced by our attack. 
We leave the formal game theoretic definitions to~\Cref{app:games}.

\subsection{Model}\label{sec:model}
\paragraph{Notation and terminology.}
For a positive integer $n$, we use $[n]$ to denote the set $\{1, \dots, n\}$. 
For a set $S$, we use $\Delta(S)$ to denote the set of all probability distributions on $S$.
For a distribution $\alpha \in S$, we use $x \leftarrow \alpha$ to denote sampling an element $x$ from $S$ according to the distribution $\alpha$.
We use the notation $x \leftarrow S$ to denote that the element $x$ is sampled from $S$ uniformly at random.
For an $n$-dimensional vector $v \in S^n$, we use the indexing notion $v_{-i}$ to denote all elements in $v$ except for the $i$th element.
Throughout the paper, we use the terms ``user'', ``miner'', ``node'', and ``player'' interchangeably to denote a participant in the blockchain protocol.

\paragraph{Blockchain model.}
We denote by \Sys{} a blockchain maintained by a set of users $\N=\{x_1,x_2,\dots,x_n\}$. 
We assume that each user $x_i \in \N$ 
has some voting \power $\vi $, which determines the probability that $x_i$ creates the next block to append to the blockchain system\footnote{
For example, the voting \power in a \PoW blockchain is the 
\cns' computational power and the voting \power
in a \PoS blockchain is \cns' stake amount,
whereas the voting power in a practical Byzantine Fault Tolerance(PBFT) blockchain 
is the existence of an approved node.}.
We normalize all voting powers in the system such that they sum to 1, i.e., $\sum_{i=1}^{n} \vi = 1$.
For simplicity, 
we assume that the number of \cns and their \power distribution 
remains constant.

We also assume a fixed reward for mining a block on the blockchain, which we will term the \emph{block reward}.
We assume the blockchain has an associated cryptocurrency $\coin$ which is the denomination used to pay out block rewards to users.
Without loss of generality, we normalize the block reward so that the reward for a block is $1$;
thus each $x_i \in \N$ expects to get a reward of $\vi$ for each unfinalized block.

Additionally, we assume the underlying blockchain \Sys{} defines a protocol that users in \Sys{} should follow, henceforth known as the \emph{honest strategy}. 
The precise specification of the honest strategy is unimportant for our purposes, only that it should not preclude any user in the blockchain from contributing to the underlying consensus protocol.

Finally, for the purposes of our attack we assume the existence of a ``general blackboard'' with read and write access, which is used to initiate and coordinate the attack.
In our work, we use use the blockchain as the blackboard and a smart contract posted on the blockchain as a message on the blackboard to initiate the attack for simplicity purposes (see~\Cref{sec:attackdetails} for more details).
Nevertheless, we stress that our attack is agnostic to the actual implementation of the blackboard, with the only requirement being that it should be publicly accessible to all users in the blockchain.
This makes our attack also applicable even for Bitcoin-like blockchain protocols which are not Turing complete and do not support general smart contracts.

\paragraph{System model.} We assume time proceeds in discrete time steps where at each time step a block is added to the blockchain. 
For simplicity purposes, we also assume that changes in the state of the blockchain and financial system (e.g., price of $\coin$) as a result of user actions at some time step $t$ will be immediately reflected in the next time step $t+1$, i.e., the consequences of user actions do not incur any time lag. 
We also assume \emph{price stability}, that is, so long as attacks remain undetectable (more on this in our cost model below), the price of the associated cryptocurrency $\coin$ would remain constant. 
Naturally, this includes the case where there are no attacks and all users are honest and adhere to the stipulated blockchain protocol.
Essentially, this assumption removes any external influence outside of the blockchain system and user strategy space on the price.

\paragraph{Cost model.}
We assume the main cost to users comes from a downward movement in the price of the cryptocurrency $\coin$, which we further assume stems from a detectable deviation from the stipulated blockchain protocol (henceforth known as a detectable attack). 
In the context of the censorship attacks involving excluding some number of users which we will present in this work, we define detectability as a step function involving the number of excluded users. 
More precisely, we assume the existence of some threshold $\eta$ of excluded users whereby a censorship attack excluding $\eta$ or more users would render the attack detectable to the excluded users and would cause the price of $\coin$ to plummet\footnote{Using a similar argument to the detectability of selfish mining and the impact on the underlying cryptocurrency in~\cite{BahraniW23}}. 
Let $\alpha>0$ be some arbitrary large constant.
We thus formally define the cost of some strategy $\sigma$ excluding $k$ amount of users as follows:
$$f(\sigma) = \begin{cases}
		0, & \text{if $k < \eta$}\\
            \alpha & \text{if $k \geq \eta$}
		 \end{cases}$$
where the first case corresponds to the cost in the undetectable setting and the second case corresponds to the cost in the detectable setting.

\paragraph{Adversarial and utility model.}
We assume all users participating in the blockchain protocol are rational, i.e., each user strives to maximize their expected utility. 
Let $\sigma_i^T  = \{\sigma_i^1, \dots, \sigma_i^T\}$ denote the $T$-block horizon mining strategy of user $x_i$, where $\sigma_i^t$ for $t\in [T]$ denotes the strategy of $x_i$ at the $t$-th step of the $T$ block interval. 
For the censorship attack presented in~\Cref{sec:attackdetails}, we assume that the only source of revenue to users is the block reward. 
Thus, the expected utility of each user $x_i$ in the first attack setting is simply $x_i$'s average block reward less any incurred costs. 
This is defined, with respect to a user $x_i$'s strategy $\sigma_i^T$ as 
$ \mathbb{E}^{\sigma_i^T}\Big[ \frac{\sum_{t=1}^T (r_t - f(\sigma_i^t)\cdot v_i)}{T} \Big] $, where $r_t$ is the $t$th block reward that $x_i$ receives.

\subsection{Attack details}\label{sec:attackdetails}

We begin with the assumptions behind our attack.

\begin{description}
\item[Rational censorship attack assumptions:]\hfill

\begin{enumerate}
\item (Honest threshold.) We assume there is a globally-known 
\powert $\threshold$ such that,
within a time period and for a set of nodes $\Nh \subset \N$ with $\sum_{i \in \Nh} \vi > \threshold$, the system \Sys{} functions perfectly if all nodes in $\Nh$ follow the honest strategy.
Furthermore, we assume that $t \geq \frac{1}{2}$
to avoid the situation where \cns split into two independent functional subsets~\cite{dwork1988consensus,GilbertL02}.
\item (Random and unknown response order.) 
We assume that the order in which nodes respond to the attack (see~\Cref{sec:attackdetails} for more details) is unknown to any of the nodes. 
More precisely, we assume the order of response is a random variable whose law is unknown and out of the control of any of the nodes.

\item (No big player.) 
No single node has voting power that exceeds $(1-t)$. 

\item (Detectability threshold.)
We assume the detectability threshold $\eta$ as defined in the cost model in~\Cref{sec:model} is known. In~\Cref{sec:eta_estimation} we describe how to realistically estimate $\eta$ in practice.
We also assume that the constant $\alpha$ representing the cost of excluding $k\geq \eta$ users (see cost model in~\Cref{sec:model} for details) is larger than any profits gained from the attack.


\end{enumerate}
\end{description}

\begin{algorithm}[t!]
\caption{$CTA$ smart contract}
\label{alg:cta}
\SetCommentSty{textnormal}
\DontPrintSemicolon
\SetKwInOut{Input}{input}
\SetKwFor{For}{for}{}{end}

\BlankLine
\Input{node powers $v_i$, node ids $x_i$}{}\;

$v:=$ declared power of user who posted the smart contract\;
$V := v$\;
$\N_a := \emptyset$\;
$\textsc{ActiveAttack} :=0$ \;
\While{new pair $\{x_i,v_i\}$ and $V<t$ and $n-|\N_a| < \eta$ and clock $<T$}{
$V := V+v_i$\;
$\N_a := \N_A \cup x_i$\;
}
\If{$V < t$ and clock $\geq T$}{
abort \; 
}
return $\textsc{ActiveAttack}=1$  \; 
\end{algorithm}

\begin{algorithm}[t!]
\caption{Rational Censorship Attack}
\label{alg:exclusion}

\SetCommentSty{textnormal}
\DontPrintSemicolon
\SetKwInOut{Input}{input}
\SetKwFor{For}{for}{}{end}
\SetKwProg{Send}{Send}{:}{}
\SetKwProg{Deliver}{Deliver}{:}{}

\BlankLine
\Send{Upon receiving a send request $m$ to $i$}{
    \If{attack is not active or $i \in \N_a$ }{
        Send $m$ to $i$
    }
}
\;
\Deliver{Upon receiving a deliver request $m$ from $i$}{
    \If{attack is not active or $i \in \N_a$ }{
        Deliver $m$ from $i$
    }
}
\end{algorithm}

The attack consists of $2$ phases: setup and attack. In the setup phase, a node creates a smart contract, known henceforth as the ``Call to Attack'' ($CTA$) smart contract. 
The contract, formally described in~\Cref{alg:cta}, waits for nodes to respond to the call with their ids and powers. 
As the nodes sequentially input their ids and powers, the creator of the contract stores the ids of the responding nodes. The contract also makes two checks: (1) the sum of the powers of the responding nodes (denoted by $V$) is larger than the power threshold $t$, and (2) the number of excluded nodes $n-|\N_a|$ is less that the detectability threshold $\eta$ (due to assumption $4$ above, it is unprofitable and hence irrational to exclude $\geq \eta$ nodes).
The contract terminates when $V \geq t$ or when some global timeout parameter $T$ is reached, whichever earlier, and outputs either success or failure respectively.
We assume that within this timeout parameter $T$ players' powers remain the same.
In the case of success, the node that creates the smart contract then disseminates the set of attack node ids $\N_a$ to all other nodes in $\N_a$.

Note that in the setup phase we require that the order in which nodes respond to the $CTA$ smart contract is random and unknown (this is needed in our proof of~\Cref{thm:declare} which shows the stability of our attack strategy). 
As we will show in~\Cref{sec:analysis}, responding positively (i.e., joining the attack) is a dominant strategy in the induced game, and thus nodes could rush input their details to the smart contract, making this random order unlikely to occur in practice.
However, this random order can still be imposed if the smart contract uses a secure random shuffle protocol with a random seed drawn from some randomness beacon (e.g., NIST~\cite{nist} or Randao~\cite{randao}) to shuffle the inputs. 

The attack phase is launched depending on a predetermined internal trigger among the attack nodes $\N_a$. 
When the attack is not yet triggered, each node in $\N_a$ sends and delivers messages normally as per the underlying blockchain protocol.
The moment the attack is triggered, each node only sends and delivers messages to and from nodes in the attack group $\N_a$, effectively censoring and excluding the remainder nodes $\N \setminus \N_a$ from participating in consensus in the blockchain.
We formally detail the attack in~\Cref{alg:exclusion}.

We note that following a successful outcome of the attack, an attacking player $x_i$ with power $v_i$ receives a reward of $\frac{v_i}{V}$.
We stress that the reward is \emph{independent} of the power value $v$ declared by $x_i$ during the call in the setup phase, as the rewards distribution mechanism functions as per the typical blockchain rewards distribution mechanism (see~\Cref{sec:model} for more details), just among a smaller set of nodes.
This \emph{rewards independence} property would play a crucial role in proving~\Cref{thm:declare} in~\Cref{sec:analysis} which states that joining the attack with an honest declaration of one's power is a Nash equilibrium in the game induced by the rational censorship attack.

\subsection{Induced game}\label{sec:formalgame}
The rational censorship attack described in the previous section induces the $n$-player, $2$-stage extensive form game $\Gamma$ depicted by the game tree $T$ in~\Cref{fig:game_tree}.
In the first round of $\Gamma$, some player, say $x_1$, decides whether or not to launch the attack and post the $CTA$ smart contract on the blockchain. 
The action space in this round is simply to decide whether to attack the blockchain or not (and if so, declare their power), and we denote by $A$ the choice to attack the blockchain.
The second round of $\Gamma$ immediately after $x_1$ posts the contract on the blockchain and is played by the remaining $(n-1)$ players (i.e., $x_2, \dots, x_n$).
These $(n-1)$ players move simultaneously in this round of $\Gamma$. 
Let $\epsilon$ denote the smallest unit of power a player can declare\footnote{E.g., the smallest denomination of the coin for Proof-of-Stake blockchain systems}. 
The action space in the second round of $\Gamma$ is $\{\bot, 0, \epsilon, 2\epsilon, \dots, 1\}$, where we use $\bot$ to denote that a node \emph{does not} take any action in a specific round.  
Any action apart from $\bot$ thus denotes an intention to participate in the attack with the action value $a \in \{0, \epsilon, 2\epsilon, \dots, 1\}$ representing the declared power of the participating node.
After all the players have made a move, we model the unknown order in which players respond to the contract by chance vertices in the game tree $T$ which randomizes the player order according to an unknown probability distribution.

Finally, we stress that the choice to model $\Gamma$ as a $2$-stage extensive form game with the second round consisting of simultaneous moves instead of an $n$-stage extensive form game where the remaining players respond sequentially after the first attack move conforms to the fact that in a real-world attack players can neither respond to the contract in a fully ordered fashion, nor have full information regarding the moves of other players (i.e., whether they have joined the attack or not). 

\begin{figure}[htb!]
    \centering
    \includegraphics[width=0.8\textwidth]{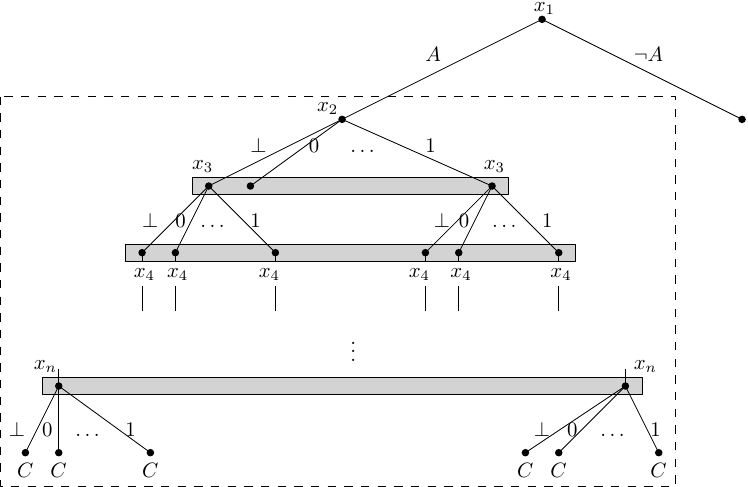}
    \caption{Game tree corresponding to $\Gamma$. The vertices in grey boxes belong in the same information set. The row of vertices labeled $C$ denote chance vertices.}
    \label{fig:game_tree}
\end{figure}

\subsection{Analysis}\label{sec:analysis}
Here we proceed to show that it is rational for players to post the attack smart contract on the blockchain as well as respond positively (i.e., agree to join the attacking coalition) to the attack.
Let $\Gamma$ be the game described in~\Cref{sec:formalgame}, and let $\subg$ be the subgame defined by the subtree in the dashed box in~\Cref{fig:game_tree} (i.e., the subtree rooted at vertex $x_2$).
Let us define the strategy profile $\subs := (v_2, \dots, v_n)$. 
That is, for players $x_i, i \in [2, \dots, n]$, the strategy $\subs_i$ for $x_i$ with power $v_i$ would be to truthfully play action $v_i$.

We define the following probabilities for each player $x_i$:

\begin{description}
\item[$p^i_1$:] $\mathbb{P}[\text{join attack and attack fails}]$
\item[$p^i_2$:] $\mathbb{P}[\text{don't join attack and attack fails}]$
\item[$p^i_3$:] $\mathbb{P}[\text{join attack and attack succeeds with player being in } \N_a]$
\item[$p^i_4$:] $\mathbb{P}[\text{don't join attack and attack succeeds}]$
\item[$p^i_5$:] $\mathbb{P}[\text{join attack and attack succeeds but player is not in } \N_a]$
\end{description}

We first show that the moment the attack smart contract is posted on the blockchain, it is rational for the remaining users respond by agreeing to be part of the attacking coalition and declaring their power truthfully.
To do so, we establish the following lemma which states the expected utility of each responding player $x_i$ given the existence of an attack and under the strategy $\subs$. 
Let $V\geq t$ denote the total power of the attacking coalition in the case of a successful attack.

\begin{lemma}\label{lem:utility}
    The expected utility of $x_i$ under the strategy $\subs$ is $p^i_3 \cdot \frac{v_i}{V}$.
\end{lemma}

\begin{proof}
Under the strategy $\subs$, all players join the attack and declare their powers truthfully. 
Thus, $u_i(\subs) = p^i_1 \cdot v_i + p^i_3 \cdot \frac{v_i}{V} + p^i_5 \cdot 0 = p^i_3 \cdot \frac{v_i}{V}$, as $p^i_1=0$ since all players join the attack under $\subs$. \qed
\end{proof}

The next theorem (proof in~\Cref{app:proof_declare}) shows that $\subs$ is a Nash Equilibrium in $\subg$.

\begin{restatable}[]{theorem}{declare}
\label{thm:declare}
$\subs$ is a Nash Equilibrium in $\subg$.
\end{restatable}

Finally, let us define the strategy profile $\sigma:=(A, \subs)$. 
In the next theorem (proof in~\Cref{app:proof_spe}), we show that $\sigma$ is a subgame perfect equilibrium in the full game $\Gamma$, which implies that it is even rational to choose to attack the blockchain.

\begin{restatable}[]{theorem}{spe}
\label{thm:spe}
    $\sigma$ is a subgame perfect equilibrium in $\Gamma$.
\end{restatable}

\paragraph{Communication complexity.} The round complexity of our attack is $2$. In the first round the initiating user posts the smart contract message. The second round is needed for all other users to respond to the smart contract. 
Importantly, none of the rounds require exchanging private messages between users and thus our attack does not require users to know the identities of other users, or depend on the existence of secure communication channels between users (as would be the case if, e.g., using MPC to coordinate the attack).



\section{Plausible countermeasures}\label{sec:deterrences}

Having described the rational censorship attack, we now expound on two plausible countermeasures to the attack.
We emphasize that none of these countermeasures can \emph{completely} prevent the attack -- they just make it less lucrative for users under certain assumptions or in specific circumstances to launch the attack.
Additionally, these defensive measures also come with several assumptions, and we will discuss these assumptions and the general efficacy of these measures in greater detail in~\Cref{sec:discussion}.

The first countermeasure is more centralized in nature and requires modifying the block reward function of the underlying blockchain protocol to give larger rewards when there are more participants in the protocol.
We note that some form of correlation of rewards with participation has been discussed in Ethereum~\cite{ethereumReward} and we detail a general technique to do so in~\Cref{app:reward}.

The second countermeasure is purely decentralized in nature and involves users launching multiple smart contracts simultaneously. 
We present a summary of this countermeasure and leave the detailed description and game theoretic analysis to~\Cref{app:simultaneous}.
We first show in~\Cref{thm:coord} in~\Cref{app:simultaneous} that in the setting with multiple attack smart contracts, simply having the knowledge that there exists an attack coalition with sufficiently many players responding to a particular smart contract makes it irrational for some player to decide to break away from the coalition.
Nevertheless, such guarantees relies on the existence of some form of coordination among the users and in real world settings such coordination might involve heavy computational or communication costs.
As such, in settings where coordination is costly, a potential deterrence would be for multiple users to launch attack smart contracts near-simultaneously. 
When faced with multiple smart contract attacks and without any form of coordination, and especially if we assume that without prior knowledge players join coalitions uniformly at random, the probability of having a winner becomes strictly lower compared to the setting where there is a coordinator and few smart contracts. 
The optimality of this countermeasure is justified if we adopt the reasonable assumption of a first-movers advantage in launching the first smart contract, which makes the optimal strategy for players to be to launch attack smart contracts as soon as possible and hence near-simultaneously.

\section{Discussion}\label{sec:discussion}
We now proceed to discuss some implications of our attack.
First, we begin with a discussion of how the detectability threshold $\eta$ might be estimated in practice as well as discuss the detectability and profit tradeoff.
Next, we expound on why such attacks can be particularly problematic for both users in the blockchain, as well as protocol designers as these attacks are, at least in theory, difficult to prevent. 
Furthermore, despite the supposed rational nature of the attack as evident in our game theoretic analysis in~\Cref{sec:analysis}, as well as the difficulty of attack prevention, we have not yet observed such attacks in practice. We therefore also outline some plausible explanations as to why this is so.


\subsection{Estimating detectability threshold}\label{sec:eta_estimation}
A crucial assumption underlying the success of our attack is global knowledge of the detectability threshold $\eta$.
Indeed our attack excludes strictly less than $\eta$ users as we assume a much higher cost of bring detected compared to the potential profits of a successful attack (assumption $4$ in~\Cref{sec:attackdetails}). 
In reality, however, this assumption of a global knowledge of the undetectability threshold might be too strong, and thus we outline two broad methods of how to estimate $\eta$ for practical purposes.

First, we observe that $\eta$ is linked to the latency parameter of the underlying blockchain's broadcast network. Indeed if network latency is large, messages might be delayed due to latency and thus nodes might appear to be excluded. Thus, a reasonable choice of $\eta$ would be some function of the underlying blockchain network latency. 
Second, we can also augment the above method by performing a statistical analysis on the variance of the number of active nodes using historical data and set $\eta$ to some reasonable function of this variance in order to make the exclusion appear statistically believable. 

Additionally we emphasize that there exists a fundamental tradeoff between detectability and profitability: excluding more nodes leads to the remaining nodes sharing a larger proportionate share of the block rewards, however it also raises suspicion among the excluded nodes when too many nodes are excluded, which could lead to nodes switching to a different blockchain and thus causing a drop in the price of the underlying cryptocurrency. We believe the analysis of the optimal number of nodes to exclude in order to  maximize long-term profits is an interesting optimization problem and we leave further exploration of this problem to future work.

\subsection{Why are such attacks problematic?}
\subsubsection{For blockchain users: tragedy of the commons}
We first recall from~\Cref{lem:utility} that the expected utility of a responding user $x_i$ under the attack strategy $\sigma$ is $p^i_3 \cdot \frac{v_i}{V}$, and the expected utility of the same user in a situation without the attack would simply be $v_i$. 
Thus, depending on the unknown probability $p^i_3$, the expected utility of the responding user $x_i$ could actually be larger if the attack did not happen compared to the expected utility in the case of a successful attack. 

Nevertheless, the results of our analysis in~\Cref{sec:analysis} imply that the attack induces a certain ``tragedy of the commons'' situation.
Although the attack could potentially make $x_i$ worse off, the potential of being the initiating user (the user that initiates the attack and first posts the attack smart contract on the blockchain), which guarantees that the player can never be worse off, as well as the threat of being censored which effectively gives a payoff of $0$, is sufficient incentive to induce all users to launch or join the attack. 
Consequently, the entire blockchain is less decentralized and less secure overall.

\subsubsection{For protocol designers: difficulty of attack prevention}\label{sec:prot_designers}

\paragraph{Difficulty of detection and futility of penalization.}
Smart contracts are public and, in a censorship-resistant blockchain, persistent by nature. 
This makes it easy in a censorship-resistant blockchain to use the very existence of an attack smart contract on the blockchain as undeniable evidence of misconduct.
Nevertheless, the key outcome of the rational censorship attack is the erosion of censorship-resistance.
This makes it trivial for the attacking coalition to rewrite history by either denying the nature of the smart contract (by say disguising the $CTA$ smart contract as some other internally-coded contract, e.g., gofundme, voting poll, etc.), or simply by creating a hard fork on the block before the block containing the smart contract, effectively erasing the existence of the smart contract from history. 

Additionally, we note that while it is possible to observe the activity of the rest of the nodes in the blockchain to determine whether they are participating in or excluded from the attack, trying to impose penalties (e.g., stake slashing for Proof of Stake blockchains,) on the attacking coalition is futile. This is again because consensus is determined by the attacking coalition and the coalition can simply ignore any imposed penalties.

\paragraph{Limited efficacy of countermeasures.}
Another factor which makes attack prevention difficult for protocol designers is that the aforementioned countermeasures described in~\Cref{sec:deterrences} are only effective under certain assumptions.
For instance, any modification of the reward function firstly has to ensure that the supply of the currency does not balloon and cause hyperinflation.
Additionally, the lack of censorship-resistance as a consequence of the attack also implies that the attacking coalition can again simply choose to ignore the specified reward function and create a hard fork which implements an independent and fixed block reward. 
Compared to ignoring penalties this could involve modifying the actual software specification of the blockchain and thus could be more difficult to actually implement. 
Nevertheless, we note that the possibility of rewriting the stipulated protocol already makes this countermeasure less reliable.

In the case of simultaneously launching multiple smart contracts, the efficacy of this countermeasure relies on the crucial assumption that there is no coordination.
In particular, this countermeasure would lose its impact when the belief that all players will join the first attack smart contract that appears on the blockchain becomes common knowledge. 

\paragraph{Existence of external attack incentives.}
A commonly cited argument why attacks on a blockchain by their own users do not occur often in practice is due to the fact that the price of the underlying cryptocurrency would drop as a result of the attack (relative to other more stable currencies like USD or stablecoins), and this could ultimately result in a larger negative impact to the attacker compared to the rewards of the attack~\cite{BahraniW23,ford2019rationality}. 
Nevertheless, recent work by~\cite{ford2019rationality} observed that if a blockchain system exists in a larger financial ecosystem, and also assuming that the blockchain system does not control too much financial power within the larger financial ecosystem, the blockchain, users of the blockchain system could have an \emph{external incentive} to attack the system in the form of short selling (i.e., betting against the value of) the underlying cryptocurrency in the external financial ecosystem. 
Thus, so long as blockchains continue to exist in a larger financial ecosystem, the argument that maintaining the security of the blockchain is incentive compatible for users of the system becomes questionable.


\subsection{Why are such attacks unobserved in practice?}
Although the aforementioned exclusion attack is clearly beneficial to attackers, it is interesting to note that such an attack is typically unobserved in practice.
In this section, we discuss several factors (in decreasing order of likelihood) to explain the reasons.

\paragraph{Attacks are undetectable.}
The primary reason for the unobservability of such attacks in practice would be the undetectability aspect of the attack. 
For instance, if the attack only excludes a small amount of nodes $k \ll \eta$, the impact of the attack would be difficult to detect and could in fact be explained away by latency or random chance. 
As mentioned in~\Cref{sec:eta_estimation}, there is an inherent tradeoff between detectability and profitability of the attack and attackers that have lower risk-tolerance could launch attacks that exclude smaller amounts of nodes compared to $\eta$ to sacrifice short-term profits but maximize undetectability in the long run.

\paragraph{Nodes are not entirely rational.}
Our analysis relies crucially on the assumption that \emph{all} nodes are rational.
However, this assumption might be too strong in practice.
To demonstrate how reality might differ from theory,
we refer to the second-price sealed-bid auction example, in which the highest bidder wins the auction but they only pay the second-highest bid.
\cite{ausubel2006lovely} shows that bidders are are incentivized to bid their true valuation of the item in this setting.
Thus, if all the bidders are rational,
this auction would reveal the true valuation of each bidder.
However, the empirical study of~\cite{rosato2024novel} shows that
this is not the case and bidders do not always bid by their true valuation of the item even though it is rational for them to do so.
In our setting, even if only a few nodes declare their voting power arbitrarily or even maliciously in the $CTA$ smart contract, rational nodes would not trust and participate in the attack, rendering such an attack unlikely to happen. 

\paragraph{Undeniable evidence.}
Due to the public nature of smart contracts, the (pseudonymous) identity of the smart contract creator is known to all participants in the blockchain protocol, making it impossible for the identity to deny launching the attack. 
This could result in varying degrees of deterrence depending on the nature of the underlying blockchain. 
For instance, in permissioned blockchains with strong identities, or in hybrid systems which are partially centralized (e.g., blockchains deployed by workplaces or between groups of companies) this together with additional penalties like threats of expulsion or report to higher authorities, could be sufficient to identify the credentials of the attacker and deter potential attacker nodes.
Even in permissionless blockchains like Bitcoin, several recent studies~\cite{Biryukov2014DeanonymisationOC,BiryukovP15,KoshyKM14} have shown that it is relatively easy and efficient to deanonymize users based on the underlying network broadcast layer.

\paragraph{Existence of larger financial ecosystem.}
As mentioned in~\Cref{sec:prot_designers},
blockchains do not exist in a vacuum but rather in a larger financial ecosystem and users could leverage external incentives to attack their blockchain.
We note, however, that whether attacks could occur or not in this setting depends on the overall alignment of these external incentives. 
Thus, the presence of a larger financial system acts as a double edged sword that provides both incentives and deterrence to attack the blockchain.



\section{Conclusion}\label{sec:conclusion}

In this work, we propose an attack on blockchain censorship resistance that utilizes the rationality assumption common in blockchain security analysis. Our attack is simple and does not require any form of coordination among users except for messages posted on a public blackboard. We show given knowledge that other users are willing to join the attack, the strategy of joining the attack and truthfully declaring one's power is a subgame perfect equilibrium in the game induced by our attack. This implies that rational users have some incentive to attack the blockchain even in the absence of external financial incentives like bribes.

We conclude our work with two interesting directions of future research. First, it would be interesting to see if one can develop countermeasures with stronger security guarantees against the attack, perhaps utilizing cryptographic techniques to hide the identities of the users.
Second, we note that our attack strategy is only one one equilibrium in the induced game and there might exist other equilibrium strategies. Thus, it would be interesting to discover other equilibrium strategies perform a comparison based on social welfare to measure how well (or rather, ``badly'') our equilibrium performs in terms of social welfare compared to other equilibria.

\subsubsection*{Acknowledgments.} This work was partially supported by MOE-T2EP20122-0014 (Data-Driven
Distributed Algorithms).


\clearpage
\bibliographystyle{splncs04}
\bibliography{reference}

\appendix

\section{Game theoretic definitions and preliminaries}\label{app:games}

\paragraph{Strategic games and Nash equilibrium.}
Let $\Gamma = (N, (A_i), (u_i))$ be an $N$ player simultaneous strategic game where $A_i$ is a finite set of actions for each player $i \in [N]$ and denote by $A := A_1 \times \cdots \times A_N$ the set of action profiles. 
The utility function of each player $i$, $u_i: A \rightarrow \mathbb{R}$, gives the payoff player $i$ gets when an action profile $a \in A$ is played. A \emph{strategy} $\sigma_i \in \Delta(A_i)$ of a player $i \in [N]$ is a distribution over all possible actions of the player. We say player $i$'s strategy is \emph{pure} if it a Dirac distribution over $A_i$.

\begin{definition} (Nash equilibrium).
A Nash equilibrium (NE) of $\Gamma$ is a product distribution $\alpha \in \times_{j \in [N]} \Delta(A_j)$ such that for every player $i \in [N]$ and for all $a'_i$ in $A_i$, 
$$\mathbb{E}_{a \leftarrow \alpha}[u_i(a)] \geq \mathbb{E}_{a \leftarrow \alpha}[u_i(a'_i, a_{-i})] $$
\end{definition}

\paragraph{Extensive form games and game trees.}
The above definition of Nash equilibrium is only applicable to single-shot games.
Games that span multiple rounds (where players' actions arrive sequentially) are modelled as extensive-form games.
An extensive-form game can be represented as a finite game tree where for every non-leaf vertex $x$ there are functions that describe the player that moves at $x$, the set of all possible actions at $x$, and for each action $a$, the child node that leads from $x$ given $a$. 
Moreover, in the imperfect information setting, or a mixed setting where players can make simultaneous moves, all player vertices are further partitioned into information sets $I$ which captures the idea that the total information about the game given to a player that makes a move at any vertex $x \in I$ is the same as making a move from any other vertex $x' \in I$.
Thus, the player is effectively rendered uncertain of their precise location in the game tree modulo the vertices in their information set.
A path from the root of the game tree to a leaf vertex corresponds to a \emph{game play} in $\Gamma$ which is a sequence of player moves made by the players in the game. 
Each leaf node is assigned a payoff vector which represents the payoff of each player if the game terminates at this leaf.
Finally, in our setting we also account for the existence of chance nodes in the game tree. 
These are nodes that are controlled by the environment and not by any player, and hence models randomness injected into the game that is beyond the control of the players.
For a formal definition of extensive-form games, see, e.g., \cite{Osborne1994}.

\paragraph{Subgame perfection.}
A subgame of an extensive-form game corresponds to a subtree rooted at any non-leaf vertex $x$ that belongs to its own information set $I$, i.e., there are no other vertices that are in $I$ except for $x$. A strategy profile is a \emph{subgame perfect equilibrium} if it is a Nash equilibrium for all subgames in the extensive-form game.

\section{Omitted proofs}\label{app:proofs}

\subsection{Proof of~\Cref{thm:declare}}\label{app:proof_declare}
\declare*

\begin{proof}
We first note that there are $3$ classes of pure strategy deviations that can be made by any player $x_i$: 
\begin{enumerate}
    \item $x_i$ plays $\subs^{'}_i = \bot$
    \item $x_i$ plays $\subs'_i = v<v_i$ 
    \item $x_i$ plays $\subs'_i = v>v_i$ 
\end{enumerate}
For each class of strategy deviations, let us denote $\subs' = [\subs'_i, \subs_{-i}]$. We will go through each class and show that for each class of strategy deviations $u_i(\subs) \geq u_i(\subs')$.
Let $s : [n-1] \rightarrow [n-1]$ be some permutation and suppose players respond to the smart contract in the following order: $\{(x_{s(i)=1}, v_{s(i)=1}), \dots, (x_{s(i)=n-1}, v_{s(i)=n-1})\}$ for $i \in \{2, \dots, n\}$.

For the first class of deviations, the expected utility of player $x_i$ is $u_i(\subs') = p^i_2 \cdot v_i + p^i_4 \cdot 0$. 
However, due to the no big player assumption and since the rest of the players still play according to $\subs$, excluding $x_i$ from the coalition would still ensure the rest of the coalition has power $>t$, thus the attack would always succeed even in the case where $x_i$ decides not to join the attack. 
Hence, $p^i_2=0$ and $u_i(\subs') = 0.$
This is no greater that the expected utility of $x_i$ under $\subs$ which is $p^i_3 \cdot \frac{v_i}{V}$ from~\Cref{lem:utility}.

Now we compute the expected utility of $x_i$ under the second class of deviations where $x_i$ joins the attack but declares some power $v<v_i$. Using the same reasoning as before, the expected utility of $x_i$ under $\subs'$ is $u_i(\subs') = p^i_3 \cdot \frac{v_i}{V'}$ where $V'$ is now the total power of the coalition in the case where the coalition succeeds under strategy $\subs'$.
The probability $p^i_3$ that $x_i$ joins the successful attack and is in the coalition remains the same in both $\subs$ and $\subs'$ as this probability in the case of smaller declared power depends purely on the random ordering of players which is agnostic to the declared powers. 
Let us denote by $k$ the index in which $\sum_{j=1}^{k-1} v_{s(i)=j} < t$ and $\sum_{j=1}^{k} v_{s(i)=j} \geq t$, and $P$ the set of players $x_{s(i)=j}$ for $j\in [k]$.
We also use $V$ to denote the total amount of power of the coalition $P$, i.e., $V=\sum_{j=1}^{k} v_{s(i)=j} $.
Now suppose $x_i$ is in the set $P$.
Under the strategy $\subs'$, if $V-(v_i-v) \geq t $, then under the rewards independence property the expected utility of $x_i$ will still be $\frac{v_i}{V}$, which is the same as the expected utility under $\subs$.
If, however, $V-(v_i-v) < t $, then additional players would have to join the coalition $P$ such that the declared powers will be at least $t$.
Let $k'$ denote the total number of players needed to ensure the sum of declared powers $\geq t$ and $V'$ denote the new sum of declared powers, i.e., $V'=\sum_{j=1}^{k'} v_{s(i)=j}$. Since $k'>k$, $ \frac{v_i}{V'} < \frac{v_i}{V}$, which makes the expected utility of $x_i$ under $\subs'$ strictly less than that under $\subs$ in this case.

We now turn to the final class of deviations where $x_i$ joins the attack but declares some power $v>v_i$.
Let us denote by $k$ the index in which $\sum_{j=1}^{k} v_{s(i)=j} < t$ and $\sum_{j=1}^{k+1} v_{s(i)=j} \geq t$, and $P$ the set of players $x_{s(i)=j}$ for $j\in [k]$.
Let $V :=\sum_{j=1}^{k+1} v_{s(i)=j}$.
Consider the positioning of $x_i$ in the random ordering of the players.
If $s(i) \leq k$, declaring $v>v_i$ with $v+\sum_{(x',v')\in P\setminus \{(x_i,v_i)\}} v' < t$ would still result in the $(k+1)$th player added to the coalition. 
From the rewards independence property, the payoff of $x_i$ in this case would be $\frac{v_i}{V}$ which is the same as declaring $v_i$ under $\subs$.
Now if $v+\sum_{(x',v')\in P\setminus \{(x_i,v_i)\}} v' \geq t$, the attack will fail as the total power of the coalition is less than the specified threshold $t$. 
If $s(i)=k+1$, $x_i$ will be accepted into the attack coalition no matter what power $x_i$ declares and the attack will succeed. Furthermore, due to the rewards independence property, the payoff of $x_i$ in this case will still be $\frac{v_i}{V}$ where $V$ is the total power of the coalition. 
If $s(i)>k+1$, the attack will still succeed but $x_i$ will not be accepted into the attack coalition no matter what power $x_i$ declares. The payoff of $x_i$ in this case is $0$. 
Altogether, we see that under the strategy $\subs'$, the probability of the attack failing is larger than that in the case of $\subs$, even though the payoffs in the case of failure or success are the same. Thus in expectation, the utility of $x_i$ under $\subs'$ is no greater than that under $\subs$.


Finally, we note that any mixed strategy must be a convex combination of any of the above described pure strategies, and since the payoff of $x_i$ under these pure deviation strategy classes is at most the payoff under $\subs$, the payoff of $x_i$ under a mixed deviation strategy must also be at most that under $\subs$. 
Hence we conclude that $\subs$ is a Nash Equilibrium in $\subg$. \qed
\end{proof}

\subsection{Proof of~\Cref{thm:spe}}\label{app:proof_spe}

\spe*

\begin{proof}
Observe that $\Gamma$ only contains $2$ subgames, the first being the full game corresponding to the game tree rooted at $x_1$, and the second being the subgame $\subg$ corresponding to the game tree rooted at $x_2$.
From~\Cref{thm:declare}, we know that $\subs$ is a Nash Equilibrium in $\subg$.
Hence, it suffices to show that no deviation from $\sigma$ can improve the expected utility of $x_1$ at round $1$ of $\Gamma$.
Under $\sigma$, $u_1(\sigma) = \frac{v_1}{V}$ where $V$ is the total power of the attacking coalition. 
Since $V<1$, this is larger than $v_1$ which is the expected utility of $x_1$ when $x_1$ chooses to deviate to the pure strategy of not launching the attack.
Due to convexity, any mixed strategy would not give an expected utility which is larger than the pure strategy of attacking. 
Therefore we conclude that $\sigma$ is a subgame perfect equilibrium in $\Gamma$. \qed
\end{proof}

\section{Attack countermeasures}
\subsection{Modifying the reward function.}\label{app:reward}
A natural countermeasure to consider is to modify the block reward function such that it is not fixed and independent of the number of participants in the blockchain, but rather a monotone increasing function of the participants.
Specifically, instead of a fixed reward of $1$ paid out to the miner of a block (see~\Cref{sec:model}), the protocol can be modified to pay out a reward of $\fn$ instead, where $\Tiln \leq n $ is the number of effective participants in the protocol. For the right choice of $\fn$, this could make censorship attacks unprofitable compared to the honest behavior of following the protocol, as censorship attacks reduce the effective number of protocol participants.

In order to ensure that censorship attacks are unprofitable, the reward function $f(\cdot)$ needs to be chosen such that for every participant $x_i$ with power $v_i$ and for $\Tiln <n$, $\frac{v_i}{V} \fn < v_i f(n) \implies \fn < t \cdot f(n)$ since $t$ is a lower bound for $V$. 
Determining the specific form of the reward function depends on the parameters $n, \Tiln$ and $t$.
For instance, for $t=\frac{1}{3}$, if we can guarantee that $n>3\Tiln$\footnote{For instance, if the distribution of power in the blockchain is completely equal, i.e., every user in the blockchain has the same amount of power.}, then the reward function $f(\cdot)$ can be any linear function. 
A reward function that is linear in the number of participants also has the added advantage that the reward can be rescaled to the original reward of $1$ per block in the case where the protocol involves the full number of participants.
Note that more unequal distributions of power would require a more careful selection of reward function, and these functions might not be rescalable\footnote{The biggest disadvantage of a non-rescalable reward function is that it will impact on the supply of the currency and hence affect the value. Analyzing this impact together with decentralised techniques of price control is an interesting but orthogonal direction that should be explored in detail in further work.}.

For this countermeasure to be effective, the underlying blockchain has to be able to determine the effective number of participants so as to determine the reward amount. 
This can be done, for instance, in Bitcoin by using either third party sources like~\cite{hashrate} to compute the total hash rate or manually estimated using the current difficulty level and average block arrival time.

\subsection{Launching several contracts simultaneously}\label{app:simultaneous}
A second countermeasure would be for several users to launch attack smart contracts simultaneously.
In contrast to the previous countermeasure of modifying the block reward function, we argue that this countermeasure can (and should, using rational arguments) be implemented by players themselves in a decentralized fashion. 
In order to analyze the efficacy of this countermeasure, we first analyze the payoffs of users in the setting where there are multiple attack smart contracts launched independently by different users. 
Suppose there are $k<n$ attack smart contracts launched on the blockchain.
For $1 \leq j \leq k$, we say the $j$th smart contract or coalition of players \emph{wins} if the sum of total powers of players joining the $j$th smart contract exceeds the requisite power threshold $t$ before the players joining all other smart contracts. 
Note that the winner is necessarily unique which comes from our constraint that $t> \frac{1}{2}$.

The presence of multiple smart contracts also induces a strategic game $\Gamma'$ played among all players. The action space of the game increases to $\{\bot^j, 0^j, \epsilon^j, \dots, 1^j\}_{j=1}^k$ where the superscript of each action denotes playing the same underlying action as in the game $\Gamma$ but with respect to coalition $j$. 
For simplicity, the analysis in this section will only focus on a reduced action space where players just choose the coalition they want to join with the addition of not participating in any attack, i.e., the remaining players choose a single action from the action space $\{\bot, 1, 2, \dots, k\}$ where $\bot$ denotes not participating in any attack and playing an action $j \in [k]$ denotes joining coalition $j$.

\begin{remark}\label{rem:joinall}
    Note that if we do not restrict the players from joining a single coalition, the best and most rational strategy for players would be to choose to join \emph{all} coalitions, as joining a coalition does not cost players anything. In this setting, the winning coalition would trivially be the coalition corresponding to the first smart contract to announce its intention to attack on the blockchain.
\end{remark}

For player $x_i$, let $p^j_i, 1 \leq j \leq k$ denote the following probability:

\begin{align*}
    p^j_i &:= \mathbb{P}[x_i \text{ joins coalition }j \land \text{coalition }j \text{ wins } \land x_i \text{ is chosen to be in the attack coalition}]\\
    &=\mathbb{P}[x_i \text{ joins coalition }j]\mathbb{P}[j \text{ wins} \mid x_i \text{ joins coalition }j] \mathbb{P}[x_i \text{ chosen } \mid j x_i \text{ wins} \land x_i \text{ joins coalition }j]
\end{align*}

We analyze the case where all but one (say $x_i$) of the players decide to join some coalition, say coalition $j$.
Let $\sigma$ denote the strategy profile where the $i$th element of $\sigma$, $\sigma_i$, represents the pure strategy that $x_i$ chooses to join coalition $j$. 

We first note in the following lemma that given that all other players are joining attack coalition $j$, any strategy for player $x_i$ that puts non-zero mass on the $\bot$ action, i.e., not participating in the attack, is necessarily dominated by any other strategy that shifts this probability mass to some other action. 

\begin{lemma}\label{lem:avoid}
    Suppose all other players play according to $\sigma_{-i}$. Let $\sigma_i$ be some strategy such that the probability mass on $\bot$ is some $\alpha >0$. Then $\sigma_i$ is always dominated by another strategy $\sigma'_i$ that puts exactly $0$ mass on $\bot$.
\end{lemma}
\begin{proof}
    For player $x_i$, the revenue of not participating in an attack is either $0$ in the case of a successful attack, or $v_i$ in the case of an unsuccessful attack. However, due to the no big player assumption, we know that the exclusion of only $x_i$ from the attack coalition $j$ does not prevent its success. Thus, $x_i$ can only gain a larger revenue by shifting the $\alpha$ mass on $\bot$ to action $j$. \qed
\end{proof}

We then show in the next theorem that, knowing that all other players are joining a specific coalition $j$, no player has any incentive to deviate from this strategy.

\begin{restatable}[]{theorem}{coord}
\label{thm:coord}
    $\sigma$ is a pure strategy Nash equilibrium in $\Gamma'$.
\end{restatable}

\begin{proof}
From~\Cref{lem:avoid} we can focus on strategies restricted to the action space $\{1, \dots, k\}$.
However, again from the no big player assumption, $p^{j'}_i = 0$ for $j'\neq j$ as $\mathbb{P}[j' \text{ wins} \mid x_i \text{ joins coalition }j']$ is necessarily $0$ since all other players are in coalition $j$. Thus, the expected revenue of $x_i$ when putting all probability mass on action $j$ is always going to be no less than the expected revenue of $x_i$ when putting nonzero mass on some other action $j'\neq j$. Thus, $\sigma$ is a pure strategy Nash equilibrium in $\Gamma'$. \qed
\end{proof}

\begin{remark}\label{rem:joinanother}
    The above theorem statement assumes all players (including the ones who launched the smart contract) can make only $1$ action, thus an implication of~\Cref{thm:coord} is that even for players who want to launch attacks, knowing that all other players intend to form some coalition is sufficient deterrence to not even consider launching a separate attack smart contract and forming a separate coalition. 
    An interesting observation is that~\Cref{thm:coord} still holds (using exactly the same analysis) in the case where we relax the action constraints and allow the players who launch the attacks \emph{an additional action} which is to decide which coalition to join (and indeed they can choose to launch an attack but still join another the coalition corresponding to someone other player's attack). 
\end{remark}

~\Cref{thm:coord} relies on the presence of some form of coordination between players.
In a real world setting, such coordination can be achieved by agreeing to join the first smart contract (which will be unique due to the fact that blockchains are totally ordered) observed on the blockchain. 
This agreement presupposes either some form of external coordination in the form of an external mediator passing on this information to all players, or can be achieved internally among the players through MPC~\cite{cramer2015secure}. 
Both of these coordination solutions, however, necessitates either some element of trust or involved heavy computational and communication cost. 

In the absence of coordination, a potential deterrence would be for multiple parties to launch smart contracts simultaneously (or near-simultaneously). 
Indeed we observe that from the analysis in~\Cref{sec:analysis}, the launcher of the smart contract is automatically included in the attack coalition and thus the main concern for the attacker would be the probability of winning.
We also note that there is a certain first-movers advantage involved in launching the smart contract -- launching first, together with adequate time between the first and subsequent attack smart contracts could influence players' beliefs in the strength of the coalition corresponding to the attack (especially if players see that sufficiently many other players have already responded to the contract). 
Altogether, this implies that players cannot lose out if they launch smart contracts as early as possible, especially if we allow them the additional move of joining another coalition that does not need to be the coalition associated with their smart contract as per~\Cref{rem:joinanother}.
In such a setting, the strategy that maximises the utility of rational players would be for all players to launching attack smart contracts as soon as possible and hence near-simultaneously.
When faced with multiple smart contract attacks and without any form coordinator, and especially if we assume that without prior knowledge players join coalitions uniformly at random, the probability of having a winner becomes strictly lower compared to the setting where there is a coordinator and few smart contracts. 
As such, we see that in such a setting the innate rational strategy of players itself could possibly serve as an natural remedy to the attack.

\end{document}